\documentclass[preprint,12pt]{elsarticle}
\usepackage{threeparttable}
\usepackage{amsmath,amssymb,amsthm}
\usepackage{multirow}
\usepackage{booktabs}
\usepackage{rotating}
\usepackage{color}
\biboptions{sort&compress,square,comma}
\makeatletter 
\@addtoreset{equation}{section}
\makeatother    
\renewcommand\theequation{{\thesection}.{\arabic{equation}}}
\allowdisplaybreaks
\newtheorem{thm}{Theorem}[section]
\newtheorem{lem}[thm]{Lemma}
\newdefinition{rmk}{Remark}
\newproof{pf}{Proof}
\allowdisplaybreaks

\begin{document}
\begin{frontmatter}

\title{A relative error estimation approach for single index model}

\author[a]{Zhanfeng Wang\corref{cor1}}
\author[a]{Zimu Chen}
\author[a]{Yaohua Wu}
\address[a]{Department of Statistics and
Finance, University of Science and Technology of China, Hefei, China.}
\cortext[cor1]{Department of Statistics and
Finance, University of Science and Technology of China, Hefei, China. (Email: zfw@ustc.edu.cn).}

\begin{abstract}
A product relative error estimation method for single index regression model is proposed
as an alternative to absolute error methods, such as the least square estimation and the least absolute deviation estimation.
It is scale invariant for outcome and covariates in the model. Regression coefficients are estimated via a two-stage procedure and their statistical properties such as consistency and normality are studied. Numerical studies including simulation and a body fat example show that the proposed method performs well.
\end{abstract}

\end{frontmatter}

\section{Introduction}

Relative errors are invariant under scale transformation of outcome and covariates in regression models, and are robust against outliers\textsuperscript{\cite{Chen:2010ic,Chen2016LPRA}}. Hence, relative error-based methods are more of concern in many applications as an alternative to absolute error methods, such as the least square estimation and the least absolute deviation estimation. About studies of the relative errors on both of theoretical and application viewpoints, please refer to [\citenum{Chen:2010ic,Chen2016LPRA,Park1998227,Ye:2007fz,zhang2012local}] and therein references.
For example, [\citenum{Chen:2010ic}] and [\citenum{Chen2016LPRA}] proposed the least absolute relative error  estimate (LARE) and the least product relative error  estimate, respectively, for a multiplicative linear regression model.  The LARE criterion was used to build a local least absolute relative error estimation method for a partially linear multiplicative model ([\citenum{zhang2012local}]). A relative error-based change point
estimation approach was proposed in [\citenum{wang2015change}].

Most of the previous studies on the relative errors, however, are built for multiplicative linear regression models and partially linear multiplicative models.
Single index models are more flexible compared to linear models and are widely used in many fields such as medicine and econometrics.
To the best of our knowledge, estimation approaches for single index models are mostly built on the absolute error in the literature, details
please referring to [\citenum{hristache2001direct,ruppert2003semiparametric,stute2005nonparametric,zhu2006empirical,Wang:2010dj,Chang:2010kk,kong2012single,hardle2012nonparametric}]
and therein references. Hence, a relative-error based approach for single index models is much desired, especially when the relative error is more of concern
in the analysis of financial data and survival data.

Since output (response) variable is generally positive when relative error is of concern, following {\cite{Chen:2010ic,Chen2016LPRA}} for the linear regression model, it generally
considers the following multiplicative single index model,
\begin{equation}
Y_i=\exp\left\{g(X_i^\tau \beta)\right\}\varepsilon_i,\qquad i=1,\cdots,n,
\label{eq:org}
\end{equation}
where $Y_i$ is response variable with positive value, $X_i$ is a $p$-vector of explanatory variables
, $\beta$ is unknown $p\times 1$ vector of regression parameters, $g$ is a unknown link function and $\varepsilon_i$ is error term with positive value. For identifiability of parameter estimation,
it is often assumed that
$\|\beta \|=1$ and the $r$th component of $\beta$ is positive,
where $\|\cdot\|$ stands for the Euclidean metric.

For model (\ref{eq:org}),
two types of relative errors: one relative to the response and the other relative to
 the predictor of the response, can be constructed as
 \begin{eqnarray*}
 \frac{Y_i-\exp\left\{g(X_i^\tau \beta{\color{black})}\right\}}{Y_i} ~~~{\rm and}~~~   \frac{Y_i-\exp\left\{g(X_i^\tau \beta)\right\}}{\exp\left\{g(X_i^\tau \beta)\right\}}.
 \end{eqnarray*}
  Only using one of the relative errors can lead to biased estimation [\citenum{Chen:2010ic}].
 Thence, in this paper  we propose the least product relative error (LPRE) criterion,
\begin{equation}
        LPRE(\beta)\equiv \sum^{n}_{i=1}\left\{
        \left| \frac{Y_i-\exp\{ g(X_i^\tau\beta)\}}{Y_i}  \right|
        \times
        \left| \frac{Y_i-\exp\{ g(X_i^\tau\beta)\}}{\exp\{ g(X_i^\tau\beta)\}}  \right|
\right\}. \label{lpre}
\end{equation}
A two-stage estimation procedure for $\beta$ and $g(\cdot)$ is developed based on the LPRE criterion,
which has the unit-free or scale-free property. As in [\citenum{Chen2016LPRA}], the LPRE method also possesses certain robustness.
We also study asymptotic properties such as consistence and normality of the proposed estimation.
Numerical studies including simulation and a body fat data analysis demonstrate that the proposed approach has better or comparable performance than
the least square method.

The rest of the paper is organized as follows. For model (\ref{eq:org}), Section 2 introduces LPRE estimation approach and studies asymptotic properties of the proposed estimators. Numerical studies including simulation and a real data application are conducted in Section 3. Technical proofs of main theorems are collected in Appendix.

\section{Methodologies}

Suppose $(X_i,Y_i), i=1,...,n$, are independent identically distributed (iid) from model \eqref{eq:org}.
By simple algebraic computation, the LPRE criterion \eqref{lpre} is rewritten as
\begin{equation}
        LPRE(\beta)= \frac 1 n \sum^{n}_{i=1}\left\{
        Y_i\exp\{-g({X}_i^\tau\beta)\}+Y_i^{-1}\exp\{g({X}_i^\tau\beta)-2\}
        \right\}.\label{eq:lpre}
\end{equation}
Since the link function $g$ is unknown, an estimator of $\beta_0$ can not be obtained directly via
minimizing (\ref{eq:lpre}) with respect to $\beta$. Instead,
we propose a two-stage estimation procedure to estimate $g(\cdot)$ and $\beta$.

\subsection{Two-stage estimation procedure}

Firstly, the unknown link function $g(\cdot)$ and its derivative $g'(\cdot)$ are estimated simultaneously by a local linear smoother method ([\citenum{fan1996local}]).
For a fixed $\beta$, let estimators of $g(z)$ and $g'(z)$ be $\hat{g}(z;\beta,h)=\hat a$ and $\hat{g}'(z;\beta,h)=\hat b$, respectively,  where $\hat a$ and $\hat b$ are minimizers of the weighted sum of squares
\begin{equation*}
  \sum^{n}_{i=1}\big\{\log Y_i-a-b(X_i^\tau\beta-z)\big\}^{2}K_h(X_i^\tau\beta-z)\; ,
\end{equation*}
where $K_h(\cdot)=K(\cdot/h)/h$ is a symmetric kernel function and $h=h_n$ is the bandwidth.
By Taylor expansion of \eqref{eq:lpre} with respect to $\log Y_i-g(X_i^\tau\beta)$,
the $LPRE(\beta)$ has the second order term $\sum_{i=1}^n(\log Y_i-g(X_i^\tau\beta))^2$.
From [\citenum{Chang:2010kk}], we have
\begin{align}
 &\hat{g}(z;\beta,h) = \sum^{n}_{i=1} W_{ni}(z,\beta,h)\log Y_i\;  \label{ghat}\\
 & \hat{g'}(z;\beta,h_1) = \sum^{n}_{i=1} \widetilde{W}_{ni}(z,\beta,h_1)\log Y_i\;  \label{dghat}
\end{align}
where
\[
W_{ni}(z,\beta,h) = \frac{K_h(X_i^\tau\beta-z)\big\{S_{n,2}(z;\beta,h)-(X^\tau_i\beta-z)S_{n,1}(z;\beta,h)\big\}}{S_{n,0}(z;\beta,h)S_{n,2}(z;\beta,h)-S_{n,1}^2(z;\beta,h)}\; ,\]
\[
\widetilde{W}_{ni}(z,\beta,h_1) = \frac{K_h(X_i^\tau\beta-z)\big\{S_{n,2}(z;\beta,h_1)-(X^\tau_i\beta-z)S_{n,1}(z;\beta,h_1)\big\}}{S_{n,0}(z;\beta,h_1)S_{n,2}(z;\beta,h_1)-S_{n,1}^2(z;\beta,h_1)}\; ,
\]
and
\[
 S_{n,r}(z;\beta,h) =\frac{1}{n} \sum^{n}_{i=1} (X^\tau_i\beta-z)^rK_h(X_i^\tau\beta-z),\quad r=0,1,2.
\]

{\color{black}\noindent
\textbf{Remark}
It is well known that convergence rate of the estimator of $\hat g (z)$ is slower than that of the estimator of $g(z)$ if the same bandwidth is used, which leads to a slower convergence rate of $\hat \beta$ than root-$n$. Following [\citenum{Wang:2010dj}] and  [\citenum{carroll1997generalized}], $h$ and $h_1$ take different values,
\begin{equation*}
h=\hat{h}_{\rm opt}\times n^{1/5}\times n^{-1/3}=\hat{h}_{\rm opt}\times n^{-2/15}\ \
{\rm and}\ \ {h}_1=\hat{h}_{\rm opt}\ ,
 \label{bandwith} 
\end{equation*}
where $\hat{h}_{\rm opt}$ is the optimal bandwidth of $\hat{g}$ from GCV.
}

Estimator of $\beta_0$ is built as follows.
For the identifiability of parameter estimation, we assume that $\|\beta\|=1$ and $\beta_{r}>0$ for $\beta=(\beta_{1},\ldots,\beta_{p})^\tau$. Let
$\beta^{(r)}=(\beta_1,\ldots,\beta_{r-1},\beta_{r+1},\ldots,\beta_p)^{\tau}$ after removing the $r$th component in $\beta$. The delete-one-component \textsuperscript{[\citenum{Yu:2002ka}]} used here is to transfer the restricted estimation function to the unrestricted estimation in the Euclidean space $R^{p-1}$. Then $\beta$ can be reparametrized as
\begin{equation}
 \beta=\beta(\beta^{(r)}) = (\beta_1,\ldots,\beta_{r-1}, (1-\|\beta^{(r)}\|^2)^{-1/2},\beta_{r+1},\ldots,\beta_{p})^\tau.
 \label{eq:BetaWithBetar}
\end{equation}
For true parameter $\beta_0=(\beta_{01},\ldots,\beta_{0p})^\tau$, $\beta^{(r)}_0$ satisfies $\| \beta_0^{(r)} \|<1$, and $\beta$ is infinitely differentiable in a neighborhood of $\beta^{(r)}_0$, where $\beta^{(r)}_0=(\beta_{01},\ldots,\beta_{0r-1}, (1-\|\beta_0^{(r)}\|^2)^{-1/2},\beta_{0r+1},\ldots,\beta_{0p})^\tau$. The Jacobian matrix of $\beta$ with respect to $\beta^{(r)}$ is defined as
\begin{equation}
        \label{eq:Jaco}
        J_{\beta^{(r)}}=\frac{\partial \beta}{\partial \beta^{(r)}}=(\gamma_1,\cdots,\gamma_p)^\tau\ ,
\end{equation}
where $\gamma_s~(1\leqslant s \leqslant p)$ satisfy that $\gamma_s=e_s$ for $1\leqslant s \leqslant r$, $\gamma_r=-(1-\| \beta^{(r)} \|^2)^{-1/2}\beta^{(r)}$, $\gamma_s = e_{s-1}$ for $r+1 \leqslant s \leqslant p$, and $e_s$ is a $(p-1)\times 1$ vector with $s$th component $1$ and the others $0$.

Let $\hat{g}(X_i^\tau\beta)=\hat{g}(X_i^\tau\beta;\beta,h)$ for simplicity of notations.
Plugging $\hat g$ in the $LPRE(\beta)$, it follows that
\begin{equation}
        \label{eq:EFM}
            D(\beta) := \sum^{n}_{i=1}\left\{
        Y_i\exp\{-\hat{g}(X_i^\tau\beta)\}+Y_i^{-1}\exp\{\hat{g}(X_i^\tau\beta)\}-2
        \right\}.
\end{equation}
It follows from \eqref{eq:BetaWithBetar} and  \eqref{eq:EFM} that $D(\beta) =D(\beta(\beta^{(r)}))$.
Thus, minimizing $D(\beta)$ with constraints $\|\beta\|=1$ and $\beta_{r}>0$ is equivalent to  minimizing $D(\beta(\beta^{(r)}))$
without constraints. Taking the first derivative of $D(\beta)$ with respect to $\beta(\beta^{(r)})$, we obtain an estimation equation for $\beta$,
\begin{equation}
        \label{eq:ESTD}
        Q_n(\hat g,\tilde \beta ^{(r)}) \equiv \frac 1 n\sum^{n}_{i=1} \left\{
        Y_i e^{-\hat{g}(X_i^\tau\beta)}
        -Y_i^{-1} e^{\hat{g}(X_i^\tau\beta)}
        \right\}
        \hat{g'}(X_i^\tau\beta)J_{\beta^{(r)}}^\tau X_i=0\; .
\end{equation}
Denote solution of (\ref{eq:ESTD}) by $\hat{\beta}^{(r)}$. From the transformation of (\ref{eq:BetaWithBetar}), we obtain an estimator of $\beta_0$, saying $\hat\beta$.

\noindent\textbf{Algorithm}

For an initial value of $\beta_0$, denoted by $\tilde \beta$ which can be obtained from the existing methods for the single index model, such as the minimum average variance estimator (MAVE).
\begin{itemize}
\item[]
Step 1:  Obtain an estimator $\hat{g}(z;\tilde\beta,h)$ of $g$ from (\ref{ghat}).
\item[]
Step 2:  Let $\hat{g}(X_i^\tau\beta)=\hat{g}(X_i^\tau\beta;\tilde\beta,h)$. Solving (\ref{eq:ESTD}) gives an estimator $\hat\beta$ of $\beta_0$. Then update $\tilde \beta=\hat\beta$.
\item[]
Step 3:  Repeat Steps 1 and 2 until some convergence criteria satisfies.
\end{itemize}
To solve (\ref{eq:ESTD}), an iteratively Newton-Raphson algorithm is employed in this paper, that is
\[\hat \beta=\tilde \beta +J_{\tilde \beta^{(r)}}B^{-1}_n(\hat g,\tilde \beta ^{(r)})Q_n(\hat g, \tilde \beta ^{(r)}),\]
where
\[ B_n(\hat g,\tilde \beta ^{(r)}) \equiv -\frac 1 n\sum^{n}_{i=1}\left\{
        Y_i e^{-\hat{g}(X_i^\tau\beta)}
        +Y_i^{-1} e^{\hat{g}(X_i^\tau\beta)}
        \right\}
        \hat{g'}^2(X_i^\tau\beta)J_{\beta^{(r)}}^\tau X_i X_i^\tau J_{\beta^{(r)}}
.\]

\subsection{Main theorems}
Let $\hat \beta$ and $\hat{g}^*$ be the final estimators of $\beta_0$ and $g$, where
\[\hat{g}^*(z) = \sum^{n}_{i=1} W_{ni}(z,\hat{\beta},h)\log Y_i.\]
The following two Theorems show asymptotic properties of $\hat \beta$ and $\hat{g}^*$.

\begin{thm}\label{thLinear}
Suppose that conditions \textup{C1}--\textup{C6} given in the Appendix hold,
then
  \[
  \sqrt{n}\big(\hat{\beta}-\beta_0\big)\xrightarrow{\;\mathcal{L}\;}
  N\big(0,J_{\beta_0^{(r)}}V^{-1}QV^{-1}J^\tau_{\beta_0^{(r)}}\big)\; ~~{\mbox{as}}~~ n\rightarrow\infty,
  \]
  where $Q=E\big\{g'(X^\tau\beta_0)^2J^\tau_{\beta_0^{(r)}}
\{(\varepsilon_i-\varepsilon^{-1}_i)X_i+E(\varepsilon_i+\varepsilon^{-1}_i)\log{\varepsilon_i}\ E(X_i|X_i^\tau\beta_0)\}
\{(\varepsilon_i-\varepsilon^{-1}_i)X_i+E(\varepsilon_i+\varepsilon^{-1}_i)\log{\varepsilon_i}\ E(X_i|X_i^\tau\beta_0)\}^\tau
J_{\beta_0^{(r)}}\big\}$, $V$ and $J_{\beta_0^{(r)}}$ are defined in condition \textup{C6}.
\end{thm}

\begin{thm}
  \label{th:thLS}
 Under conditions \textup{C1}, \textup{C2(i)}, \textup{C3(i)}, \textup{C4} and \textup{C5} in Appendix, and $\|\hat{\beta} -\beta_0\|=O_p\big(n^{-1/2}\big)$, we have
  \[
  \sup_{x\in A} |\hat{g}^*(x^\tau\hat{\beta})-g(x^\tau\beta_0)|
  =O_p\big( (\log n/nh)^{1/2}\big),\]
  where $A$ is defined in \textup{C1(i)}. Furthermore with bandwidth $h = O(n^{1/5})$, we have
  \[
  \sup_{(x, \beta) \in \mathcal{A}_n} |\hat{g}^*(x^\tau{\beta})-g(x^\tau\beta_0)|
  =O_p\big( n^{-2/5}(\log n)^{1/2}\big),\]
 where ${\cal A}_n=\{(x,\beta): (x,\beta)\in A\times R^p,
\|\beta-\beta_0\|\leq cn^{-1/2}\}$  for a constant $c>0$.

\end{thm}

\section{Numerical studies}
\subsection{Simulation studies}

Simulation studies are conducted to evaluate performance of the proposed estimation method (saying LPRE). Data are generated from the following model
\[Y=\exp\{\sin(2X^\tau\beta_0)+2\exp(X^\tau\beta_0)\}\varepsilon,\]
where 
$\beta_0=(1,-1,1)^\tau/{\sqrt{3}}$, the covariate $X$ is from multivariate normal distributions 
$N(0,I_3)$, and $I_p$ is identity matrix with rank $p$. We consider three different error distributions: $\log(\varepsilon)\sim N(0,1)$ and uniform distribution $U(-2,2)$; and $\varepsilon$ follows the distribution with density function, $f_{\rm eff}(x)=c\exp(-x-x^{-1}+2-\log x)I(x>0)$, where $c=0.5941$ is the normalization constant. Under $f_{\rm eff}(x)$, the LPRE estimator is efficient. Sample size $n$ takes $50$ and $100$.
All results are based on $500$ replications.

For comparison, we also consider a linear LPRE estimation (Linear) with criterion function
\[D_1^*(\beta)=\sum^{n}_{i=1}\left\{
        Y_i\exp(-{X}_i^\tau\beta)+Y_i^{-1}\exp({X}_i^\tau\beta)-2
        \right\},\]
and the least square estimation (LS) with criterion function
\[D_2^*(\beta)=\sum^{n}_{i=1}\left\{
       \log(Y_i)-{g}(X^\tau\beta)
        \right\}^2.\]
Averages of estimation error (EE) and mean squared errors (MSE) of the estimated $\hat{\beta}$
\[EE=\sqrt{|1-\hat{\beta}^\tau\beta_0|}\ ,\qquad MSE=(\hat{\beta}-\beta_0)^\tau(\hat{\beta}-\beta_0),\]
are computed to assess the accuracy of estimate $\hat{\beta}$.
And average squared error (ASE)
\[ASE=\frac{1}{n}\sum_{i=1}^n\{g(X^\tau_i\beta_0)-\hat{g}(X^\tau_i\hat{\beta})\}^2\ \]
is used to measure performance of $\hat{g}$.

Table \ref{tab1} presents averages of bias (Bias), standard error (SE) and root-mean-square error (RMSE) of estimate $\hat{\beta}$, and Table \ref{tab2} shows EE and MSE for $\hat{\beta}$, and ASE for $\hat{g}$. Since Linear does not involve the function $g$, Table \ref{tab2} does not list ASE for Linear.
From Tables \ref{tab1} and \ref{tab2}, Linear has the largest Bias, SE, RMSE, EE and MSE.
Biases are close to 0 for LS and LPRE. However, LPRE has smaller or comparable SE, RMSE , EE, MSE and ASE than LS.
Under normal errors, LPRE has comparable performance than LS. But for uniform errors and $f_{\rm eff}(x)$, LPRE performs better compared to LS.

\begin{table}[h!]
\tabcolsep=3pt\fontsize{8}{12}
\selectfont
\begin{center}
\caption{Simulation results of the bias (Bias), standard error (SE) and root-mean-square error (RMSE) of estimator $\hat{\beta}$.}
\label{tab1}
\vskip 0.3cm
{
    \begin{tabular}{*{2}{l}*{9}{c}}\toprule

 & & \multicolumn{3}{c}{$\beta_1$}&\multicolumn{3}{c}{$\beta_2$}&\multicolumn{3}{c}{$\beta_3$}\\
\cmidrule{3-11}
        $n$ & Method
        & Bias & SE & RMSE & Bias & SE & RMSE & Bias & SE & RMSE\\
          \cmidrule{1-11}
$50$ & \multicolumn{3}{l}{$\log\varepsilon \sim N(0,1)$} &&&&&&&\\
&LS & -0.4081 & 6.1700 & 6.1835 & 0.3220 & 5.9135 & 5.9223 & -0.2120 & 5.9518 & 5.9556\\
&LPRE & -0.4624 & 6.4256 & 6.4422 & 0.4255 & 6.1361 & 6.1508 & -0.1372 & 6.2450 & 6.2465\\
&Linear & 137.7814 & 71.1255 & 155.0566 & -139.3865 & 64.5025 & 153.5877 & 136.4566 & 69.5304 & 153.1498\\
& \multicolumn{3}{l}{$\log\varepsilon \sim U(-2,2)$} &&&&&&&\\
&LS & -0.2816 & 6.8547 & 6.8605 & 0.8139 & 6.7896 & 6.8382 & -0.1481 & 7.0528 & 7.0543\\
&LPRE & -0.2145 & 6.1522 & 6.1559 & 0.7678 & 6.1695 & 6.2171 & -0.0204 & 6.2632 & 6.2632\\
&Linear & 139.3320 & 68.2595 & 155.1540 & -136.9369 & 68.0073 & 152.8944 & 137.9552 & 67.7995 & 153.7154\\
& \multicolumn{3}{l}{$\varepsilon \sim f_{\rm eff}(x)$} &&&&&&&\\
&LS & 0.1515 & 3.9036 & 3.9065 & 0.0456 & 3.9420 & 3.9423 & -0.5096 & 3.9442 & 3.9770\\
&LPRE & 0.1307 & 3.9113 & 3.9135 & 0.0538 & 3.9296 & 3.9300 & -0.4794 & 3.9356 & 3.9647\\
&Linear & 138.4750 & 65.2134 & 153.0625 & -140.3098 & 64.0498 & 154.2375 & 139.6475 & 64.5872 & 153.8601\\
$100$ & \multicolumn{3}{l}{$\log\varepsilon \sim N(0,1)$} &&&&&&&\\
&LS & 0.2416 & 3.9093 & 3.9168 & 0.4594 & 4.1282 & 4.1536 & -0.2042 & 4.0129 & 4.0181\\
&LPRE & 0.2051 & 4.1420 & 4.1470 & 0.4313 & 4.3439 & 4.3653 & -0.2453 & 4.2583 & 4.2654\\
&Linear & 141.1604 & 50.2338 & 149.8322 & -138.6525 & 51.4673 & 147.8966 & 135.7281 & 49.1377 & 144.3490\\
& \multicolumn{3}{l}{$\log\varepsilon \sim U(-2,2)$} &&&&&&&\\
&LS & -0.4725 & 4.7664 & 4.7897 & 0.1979 & 4.6621 & 4.6663 & 0.0986 & 4.6164 & 4.6175\\
&LPRE & -0.4026 & 4.0799 & 4.0997 & 0.1174 & 4.0059 & 4.0076 & 0.1034 & 3.9013 & 3.9026\\
&Linear & 138.9521 & 49.2139 & 147.4100 & -137.6824 & 49.0441 & 146.1567 & 138.4797 & 48.8258 & 146.8352\\
& \multicolumn{3}{l}{$\varepsilon \sim f_{\rm eff}(x)$} &&&&&&&\\
&LS & -0.1913 & 2.5855 & 2.5926 & -0.1102 & 2.5959 & 2.5982 & -0.0927 & 2.5655 & 2.5672\\
&LPRE & -0.2051 & 2.5544 & 2.5626 & -0.1197 & 2.5745 & 2.5773 & -0.0844 & 2.5301 & 2.5315\\
&Linear & 136.5872 & 44.0885 & 143.5265 & -141.0013 & 46.4743 & 148.4629 & 134.3737 & 45.8908 & 141.9939\\
    \bottomrule
\multicolumn{8}{l}{* All results should $\times 10^{-2}$}
\end{tabular}}
\end{center}
\end{table}

\begin{table}[h!]
\tabcolsep=3pt\fontsize{8}{12}
\selectfont
\begin{center}

\caption{Simulation results of EE and MSE for $\hat{\beta}$, and ASE for $\hat{g}$.}
\label{tab2}
\vskip 0.3cm
{
    \begin{tabular}{*{1}{l}*{4}{c}*{3}{c}}\toprule
    &\multicolumn{3}{c}{$n = 50$} && \multicolumn{3}{c}{$n = 100$}\\
    \cline{2-4} \cline{6-8}
        Method
        & EE & MSE & ASE && EE & MSE & ASE \\
\cmidrule{1-8}
\multicolumn{3}{l}{$\log\varepsilon \sim N(0,1)$} &&&&&\\
LS & 6.4555 & 1.0878 & 4.8990 & &4.3424 & 0.4874 & 2.2225\\
LPRE & 6.7353 & 1.1835 & 5.3043 && 4.5656 & 0.5445 & 2.4785\\
Linear & 152.0497 & 710.8659 & --- && 153.5368 & 651.5973 & ---\\
\multicolumn{3}{l}{$\log\varepsilon \sim U(-2,2)$} &&&&&\\
LS & 7.4633 & 1.4359 & 6.7380 & &5.0761 & 0.6604 & 3.0440\\
LPRE & 6.6663 & 1.1578 & 5.4854 && 4.3226 & 0.4810 & 2.2325\\
Linear & 152.2977 & 710.7787 & --- && 153.6529 & 646.5207 & ---\\
\multicolumn{3}{l}{$\varepsilon \sim f_{\rm eff}(x)$} &&&&&\\
LS & 4.2350 & 0.4662 & 1.9161 && 2.8072 & 0.2006 & 0.8588\\
LPRE & 4.2288 & 0.4648 & 1.9075 && 2.7750 & 0.1962 & 0.8417\\
Linear & 153.3506 & 708.9026 & --- && 153.1695 & 628.0336 & ---\\
    \bottomrule
\multicolumn{7}{l}{* All results should $\times 10^{-2}$}
\end{tabular}}
\end{center}
\end{table}

\subsection{Real data example}
The proposed method is applied to a body fat data which aims to study relationship between percentage of body fat (response $Y$) and other measurement indices such as body mass index (BMI=weight/height$^2$, $x_1$), chest ($x_2$), hip ($x_3$), circumference of the skinfold measurements neck ($x_4$), age ($x_5$), abdomen ($x_6$), thigh ($x_7$), knee ($x_8$), ankle ($x_9$), biceps ($x_{10}$), forearm($x_{11}$) and wrist ($x_{12}$).
There are 252 observed subjects. Deleting two possible outliers (case number 42 and 183) one with response $Y=0$ and one with unreasonable measurement of neck, the left 250 samples are partitioned into two parts:
the first 200 samples to build the model (\ref{eq:org}) and the last 50 ones to test the fitted model.

Due to bad performance of the method Linear in simulation studies, we only present results of parameter estimation in Table 3 for the methods LS and LPRE.
In Table 3, $p$-value is calculated with $1-\Phi(|\hat{\beta}_j/\hat{s}_j|)$ for estimator $\hat{\beta}_j$ of the $j$th component of $\beta_0$, where $\hat{s}_j$ is the estimated standard deviation for $\hat{\beta}_j$ and $\Phi(\cdot)$ is the cumulative distribution function of the standard normal distribution. The variance estimation from the methods LS and LPRE are obtained by bootstrap approach. We can see that these 2 methods
identify some common variables with same estimated directions, such as BMI, age, abdomen, thigh and wrist, when nominal significance level is $0.05$.
However, for knee and ankle, parameter estimates from LPRE have positive values while those for LS are negative, although they are not significantly
equal to $0$. The LPRE shows that the percentage of body fat becomes greater as knee and ankle are larger, while the LS provides the opposite effect.

Performance of prediction from LS and LPRE is measured by four different median indices: median of absolute prediction errors $\{|Y_i-\hat{Y}_i|\}$ (MPE), median of product relative prediction errors $\{|Y_i-\hat{Y}_i|^2/(Y_i\hat{Y}_i\}$ (MPPE), median of additive relative prediction errors $\{|Y_i-\hat{Y}_i|/Y_i+|Y_i-\hat{Y}_i|/\hat{Y}_i\}$ (MAPE) and median of squared prediction errors $\{(Y_i-\hat{Y}_i)^2\}$(MSPE), where $\hat{Y}_i = \exp\{\hat{g}(X_i^\tau\hat\beta)\},\ i=201,\cdots,250$. Table 4 lists results of MPE, MPPE, MAPE and MSPE. We can see that the proposed method LPRE has
smaller MPE, MPPE, MAPE and MSPE, compared to the absolute error method LS.

\begin{table}[!h]
\begin{center}
\caption{Parameter estimates of the body fat from LS and LPRE.}
\label{tab3}
\vskip 0.3cm
\scalebox{0.8}{
    \begin{tabular}{*{5}{c}}\toprule
    & \multicolumn{2}{c}{LS}     & \multicolumn{2}{c}{LPRE}\\
    \cmidrule(l){2-3} \cmidrule(l){4-5}
    & Est.(SE)& $p$-Value & Est.(SE)& $p$-Value\\
    \cmidrule{2-5}
BMI     & 0.4706  (0.1480) & 0.0007 & 0.2566 (0.1520) & 0.0457 \\
chest   & -0.0681 (0.1230) & 0.2898 & 0.0075 (0.1231) & 0.4756 \\
hip     & -0.4505 (0.1280) & 0.0002 & -0.5241(0.1289) & 0.0000 \\
neck    & -0.2163 (0.1149) & 0.0298 & -0.2468(0.1147) & 0.0157 \\
age     & 0.1552  (0.0882) & 0.0393 & 0.2019 (0.0886) & 0.0114 \\
abdomen & 0.6518  (0.1197) & 0.0000 & 0.6482 (0.1211) & 0.0000 \\
thigh   & 0.2206  (0.1317) & 0.0469 & 0.3028 (0.1277) & 0.0089 \\
knee    & -0.0517 (0.1067) & 0.3139 & 0.0170 (0.1082) & 0.4376 \\
ankle   & -0.0008 (0.0795) & 0.4959 & 0.0067 (0.0811) & 0.4671 \\
biceps  & 0.0321  (0.1069) & 0.3819 & 0.0700 (0.1103) & 0.2629 \\
forearm & 0.1161  (0.1208) & 0.1682 & 0.0881 (0.1219) & 0.2349 \\
wrist   & -0.0971 (0.1151) & 0.1995 & -0.1816(0.1132) & 0.0543\\
\bottomrule
\end{tabular}}
\end{center}
\end{table}

\begin{table}[!h]
\begin{center}
\caption{Results of prediction from LS and LPRE.}
\label{tab4}
\vskip 0.3cm
\scalebox{0.8}{
    \begin{tabular}{*{5}{c}}\toprule
    & MPE     & MPPE & MAPE & MSPE\\
LS & 4.649 & 0.012 & 0.501 & 21.614\\
LPRE & 4.237 & 0.009 & 0.413 & 17.955\\
      \bottomrule
\end{tabular}}
\end{center}
\end{table}

\section{Conclusion}

For the multiplicative single index model, we use the two-stage procedure to propose a relative error estimation approach
which is invariant in scale of the response and the covariate variables.
The asymptotical properties: consistency and normality of the resulting estimators, are provided.
Simulation studies show that the relative error methods have comparable or better performance on estimations of the regression coefficients and
the link function, compared to the least square estimation.

\section*{Appendix}
\renewcommand{\appendixname}{\bfseries A}
\renewcommand\theequation{A.{\arabic{equation}}}

Before proving the asymptotic properties of the estimators, we need to provide the following conditions:

\renewcommand{\labelenumi}{{C\arabic{enumi}}}
\begin{enumerate}
  \item
    \begin{enumerate}[(i)]
      \item The distribution of $X$ has a compact support set $A$.
      \item The density function $f_{\beta}(\cdot)$ of $X^\tau\beta$ is positive and satisfies a Lipschitz condition of order $1$ for $\beta$ in a neighborhood of $\beta_0$. Further, $\beta$ has a positive and bounded density function $f_{\beta_0}(\cdot)$ on its support $\mathcal{T}$.
    \end{enumerate}
  \item
    \begin{enumerate}[(i)]
      \item The second derivative of function $g$ is bounded and continuous.
      \item $l_s(\cdot)$ satisfies Lipschitz condition of order 1, where $l_s(\cdot)$ is the $s$th component of $l(\cdot)$, where $l(z)=E(X|X^\tau\beta_0=z)$, $1 \leqslant s \leqslant p$.
    \end{enumerate}
  \item
    \begin{enumerate}[(i)]
      \item The kernel $K$ is a bounded, continuous and symmetric probability density function, satisfying
        \[ \int_{-\infty}^{\infty}u^2K(u)\,du \neq 0\; , \quad
            \int_{-\infty}^{\infty}|u|^2K(u)\,du < \infty\, ;\]
      \item $K$ satisfies Lipschitz conditions on ${R}^1$.
    \end{enumerate}
    \item
    \begin{enumerate}[(i)]
      \item $h \rightarrow 0$, $nh^2/\log ^2 n \rightarrow \infty$, $\limsup_{n \rightarrow \infty}nh^5 \leqslant c < \infty$;
      \item $h_1 \rightarrow 0$, $nh^2h_1^3/\log ^2 n \rightarrow \infty$, $nh^4 \rightarrow 0$, $\limsup_{n \rightarrow \infty}nh_1^5 < \infty$
    \end{enumerate}
  \item
    \begin{enumerate}[(i)]
      \item $E(\log\varepsilon|X)=0$, $\text{var}(\log\varepsilon|X)=\sigma^2<\infty$, $E\big( (\log\varepsilon)^4|X\big)<\infty$.
      \item $E(\varepsilon-\varepsilon^{-1}|X)=0$
    \end{enumerate}

  \item \quad $V=E\{(\varepsilon_i+\varepsilon^{-1}_i)g'(X_i^\tau\beta_0)^2 J_{\beta^{(r)}_0}^\tau X_iX_i^\tau J_{\beta^{(r)}_0}\}$ is a positive definite matrix, where $J_{\beta^{(r)}_0}^\tau$ is defined by \eqref{eq:Jaco};
%
\end{enumerate}
\renewcommand{\labelenumi}{\arabic{enumi}.}

Conditions C1 -- C4 are commonly used in studies of single index model.
C5 and C6 are conditions for relative errors, which guarantees the asymptotic normality of parameter estimator.

\begin{lem}
  Suppose that conditions \textup{C1}-\textup{C4} and \textup{C5(i)} hold. We then have
  \begin{equation*}
  \sup_{(x,\beta)\in \mathcal{A}_n} |g(x^\tau\beta_0)-\hat{g}(x^\tau\beta;\beta,h)|
  =O_p\big( (\log n/nh)^{1/2})\big)\; .
  \end{equation*}
  If in addition, \textup{C5}(ii) also holds, then we have
  \begin{equation*}
  \sup_{(x,\beta)\in \mathcal{A}_n} |g'(x^\tau\beta_0)-\hat{g'}(x^\tau\beta;\beta,h_1)|
  =O_p\big( (\log n/nh_1^3)^{1/2})\big)\; .
  \end{equation*}
 where ${\cal A}_n=\{(x,\beta): (x,\beta)\in A\times R^p,
\|\beta-\beta_0\|\leq cn^{-1/2}\}$  for a constant $c>0$.
\end{lem}
\begin{proof}
Proof of Lemma A.1 is similar to that of Lemma 4 in [\citenum{Wang:2010dj}]. We omit it here.
\end{proof}
\begin{lem}
  Suppose that conditions \textup{C1}-\textup{C6} are satisfied. Then we have
\begin{equation*}
    \sup_{\beta^{(r)}\in \mathcal{B}_n}
    \| R(\beta^{(r)})-U(\beta^{(r)}_0)+nV(\beta^{(r)}-\beta^{(r)}_0)\|
    =o_P\big(\sqrt{n}\,\big) \; ,
\end{equation*}
  where $\mathcal{B}_n=\{\beta^{(r)}:\| \beta^{(r)}-\beta^{(r)}_0 \| \leqslant  cn^{-1/2}\}$ or a constant $c > 0$, $V$ is defined in condition \textup{C6}, and
  \begin{align*}
    R(\beta^{(r)}) =
    & \sum^{n}_{i=1} \left\{
        Y_i e^{-\hat{g}(X_i^\tau\beta;\beta(\beta^{(r)}),h)}
        -Y_i^{-1} e^{\hat{g}(X_i^\tau\beta;\beta(\beta^{(r)}),h)}
        \right\}\hat{g'}(X_i^\tau\beta;\beta(\beta^{(r)}),h_1)J_{\beta^{(r)}}^\tau X_i\; ;\\
  U(\beta^{(r)}_0)=
      & \sum^{n}_{i=1}g'(X_i^\tau\beta_0)J_{\beta^{(r)}_0}^\tau[(\varepsilon_i-\varepsilon^{-1}_i)X_i+E(\varepsilon_i+\varepsilon^{-1}_i)\log{\varepsilon_i}\ E(X_i|X_i^\tau\beta_0)]\\
    =:\; &U_1(\beta^{(r)}_0)+U_2(\beta^{(r)}_0)\;\\
    U_1(\beta^{(r)}_0)=&\sum^{n}_{i=1}(\varepsilon_i-\varepsilon^{-1}_i)g'(X_i^\tau\beta_0)J_{\beta^{(r)}_0}^\tau X_i\\
    U_2(\beta^{(r)}_0)=&\sum^{n}_{i=1}E(\varepsilon_i+\varepsilon^{-1}_i)g'(X_i^\tau\beta_0)J_{\beta^{(r)}_0}^\tau \log{\varepsilon_i}\ E(X_i|X_i^\tau\beta_0)
  \end{align*}
\end{lem}

\begin{proof}
  Separating $R(\beta^{(r)})$, we have $R(\beta^{(r)}) = R_1(\beta^{(r)})+R_2(\beta^{(r)})+R_3(\beta^{(r)})$
    \begin{align*}
R_1(\beta^{(r)})=&\sum^{n}_{i=1}  \Bigg\{\varepsilon_i \Big [e^{g(X_i^\tau\beta_0)-\hat{g}(X_i ^\tau\beta;\beta(\beta^{(r)}),h)}-1-\big(g(X_i^\tau\beta_0)-\hat{g}(X_i ^\tau\beta;\beta(\beta^{(r)}),h)\big)\Big ]\\
      &-\varepsilon^{-1}_i \Big [e^{-[g(X_i^\tau\beta_0)-\hat{g}(X_i ^\tau\beta;\beta(\beta^{(r)}),h)]}-1+ \big(g(X_i^\tau\beta_0)-\hat{g}(X_i ^\tau\beta;\beta(\beta^{(r)}),h)\big)\Big ]\Bigg\}\\
      &\quad\times \hat{g'}(X_i^\tau\beta;\beta(\beta^{(r)}),h_1)J_{\beta^{(r)}}^\tau X_i\\
R_2(\beta^{(r)})=&\sum^{n}_{i=1} (\varepsilon_i-\varepsilon^{-1}_i)\hat{g'}(X_i^\tau\beta;\beta(\beta^{(r)}),h_1)J_{\beta^{(r)}}^\tau X_i\\
R_3(\beta^{(r)})=&\sum^{n}_{i=1} (\varepsilon_i+\varepsilon^{-1}_i)\big[g(X_i^\tau\beta_0)-\hat{g}(X_i ^\tau\beta;\beta(\beta^{(r)}),h)\big]\hat{g'}(X_i^\tau\beta;\beta(\beta^{(r)}),h_1)J_{\beta^{(r)}}^\tau X_i
    \end{align*}
For $R_1(\beta^{(r)})$, by Taylor's expansion, we have $R_1(\beta^{(r)})=R_{11}(\beta^{(r)})+R_{12}(\beta^{(r)})$
\begin{align*}
  R_{11}(\beta^{(r)})=& \sum^{n}_{i=1}(\varepsilon_i-\varepsilon^{-1}_i)\big[\hat{g}(X_i ^\tau\beta;\beta(\beta^{(r)}),h)-g(X_i^\tau\beta_0)\big]^2\hat{g'}(X_i^\tau\beta;\beta(\beta^{(r)}),h_1)J_{\beta^{(r)}}^\tau X_i\\
  R_{12}(\beta^{(r)})=& \sum^{n}_{i=1}(\varepsilon_i-\varepsilon^{-1}_i)\hat{g'}(X_i^\tau\beta;\beta(\beta^{(r)}),h_1)J_{\beta^{(r)}}^\tau X_io_p\big(\big[\hat{g}(X_i ^\tau\beta;\beta(\beta^{(r)}),h)-g(X_i^\tau\beta_0)\big]^2\big)
\end{align*}
So we only need to proof
\begin{equation*}
  \sup_{\beta^{(r)}\in \mathcal{B}_n}\left\|R_{11}(\beta^{(r)})\right\|=o_P\big(\sqrt{n}\,\big)
\end{equation*}

Separating $R_{11}(\beta^{(r)})$, we have $R_{11}(\beta^{(r)}) = H_1(\beta^{(r)})+H_2(\beta^{(r)})$
   \begin{align*}
      H_1(\beta^{(r)})=&\sum^{n}_{i=1}(\varepsilon_i-\varepsilon^{-1}_i)\big[g(X_i^\tau\beta_0)-\hat{g}(X_i ^\tau\beta;\beta(\beta^{(r)}),h)\big]^2\\
      &\quad\times\big[\hat{g'}(X_i^\tau\beta;\beta(\beta^{(r)}),h_1)-g'(X_i^\tau\beta_0)\big]J_{\beta^{(r)}}^\tau X_i \\
      H_2(\beta^{(r)})=&\sum^{n}_{i=1}(\varepsilon_i-\varepsilon^{-1}_i)\big[g(X_i^\tau\beta_0)-\hat{g}(X_i ^\tau\beta;\beta(\beta^{(r)}),h)\big]^2 g'(X_i^\tau\beta_0)J_{\beta^{(r)}}^\tau X_i \\
    \end{align*}

Note that $J_{\beta^{(r)}}-J_{\beta_0^{(r)}}=O_p\big(n^{-1/2}\big)$ for all $\beta^{(r)}\in \mathcal{B}_n$. By Lemma A.1 and by the law of large numbers we obtain

\begin{equation}
  \left\|H_1(\beta^{(r)})\right\|=o_P\big(\sqrt{n}\,\big).
\end{equation}

We can show that

\begin{equation}
  \begin{aligned}
  \frac{1}{\sqrt{n}}\left\|H_2(\beta^{(r)})\right\|=&O_p\left( \frac{\log n}{nh_1^3}\cdot \frac{1}{n} \sum^{n}_{i=1} (\varepsilon_i-\varepsilon^{-1}_i)g'(X_i^\tau\beta_0)J_{\beta^{(r)}_0}^\tau X_i \right)\\
     \xrightarrow{p}&\ 0\; ,
  \end{aligned}
\end{equation}

Substituting (A.6)--(A.9) to $R_1(\beta^{(r)})$, we prove
\begin{equation}
 \sup_{\beta^{(r)}\in{\mathcal B}_n}\|R_{1}(\beta^{(r)},\beta^{(r)})\|=o_P\big(\sqrt{n}\,\big)
\end{equation}

For $R_2(\beta^{(r)})$, Separating $R_2(\beta^{(r)})$, we have $R_2(\beta^{(r)}) =R_{21}(\beta^{(r)})+R_{22}(\beta^{(r)})$
   \begin{align*}
      R_{21}(\beta^{(r)})=&\sum^{n}_{i=1}(\varepsilon_i-\varepsilon^{-1}_i)\big[\hat{g'}(X_i^\tau\beta;\beta(\beta^{(r)}),h_1)-g'(X_i^\tau\beta_0)\big]J_{\beta^{(r)}}^\tau X_i \\
      R_{22}(\beta^{(r)})=&\sum^{n}_{i=1}(\varepsilon_i-\varepsilon^{-1}_i)g'(X_i^\tau\beta_0)J_{\beta^{(r)}}^\tau X_i
   \end{align*}

Similar to the proof of $H_2(\beta^{(r)})$, we can prove that $\|R_{21}(\beta^{(r)})\|=o_P\big(\sqrt{n}\,\big)$. Note that $J_{\beta^{(r)}}-J_{\beta_0^{(r)}}=O_p\big(n^{-1/2}\big)$ for all $\beta^{(r)}\in \mathcal{B}_n$, we have

\begin{equation}
  \sup_{\beta^{(r)}\in {\cal B}_n}\|R_{22}(\beta^{(r)})-U_1(\beta^{(r)}_0)\| = o_P(\sqrt{n}\,).
\end{equation}

%

Separating $R_3(\beta^{(r)})$, we have $R_3(\beta^{(r)}) =R_{31}(\beta^{(r)})-R_{32}(\beta^{(r)})-R_{33}(\beta^{(r)})+R_{34}(\beta^{(r)})-R_{35}(\beta^{(r)})$
   \begin{align*}
      R_{31}(\beta^{(r)})= &\sum^{n}_{i=1}(\varepsilon_i+\varepsilon^{-1}_i)\big[g(X_i^\tau\beta_0)-\hat{g}(X_i^\tau\beta;\beta(\beta^{(r)}),h)\big]\\
      &\quad \quad \times\big[\hat{g'}(X_i^\tau\beta;\beta(\beta^{(r)}),h_1)-g'(X_i^\tau\beta_0)\big]J_{\beta^{(r)}}^\tau X_i \\
      R_{32}(\beta^{(r)})=&\sum^{n}_{i=1}(\varepsilon_i+\varepsilon^{-1}_i-E\varepsilon_i-E\varepsilon^{-1}_i)\big[\hat{g}(X_i^\tau\beta_0)-\hat{g}(X_i^\tau\beta;\beta(\beta^{(r)}),h)\big]g'(X_i^\tau\beta_0)J_{\beta^{(r)}}^\tau X_i\\
      R_{32}(\beta^{(r)})=&\sum^{n}_{i=1}E(\varepsilon_i+\varepsilon^{-1}_i)g'(X_i^\tau\beta_0)J_{\beta^{(r)}}^\tau \log{\varepsilon_i}\ l(X_i^\tau\beta_0)\\
      R_{34}(\beta^{(r)})=&\sum^{n}_{i=1}(\varepsilon_i+\varepsilon^{-1}_i)\big[\hat{g}(X_i^\tau\beta_0)-\hat{g}(X_i^\tau\beta;\beta(\beta^{(r)}),h)\big]g'(X_i^\tau\beta_0)J_{\beta^{(r)}}^\tau X_i \\
      R_{35}(\beta^{(r)})=&\sum^{n}_{i=1}E(\varepsilon_i+\varepsilon^{-1}_i)g'(X_i^\tau\beta_0)J_{\beta^{(r)}}^\tau\\
      &\quad \times\Big\{X_i\big[\hat{g}(X_i^\tau\beta_0)-\hat{g}(X_i^\tau\beta;\beta(\beta^{(r)}),h)\big]-\log{\varepsilon_i}\ l(X_i^\tau\beta_0)\Big\}
   \end{align*}
where $l(X_i^\tau\beta_0)=E(X_i|X_i^\tau\beta_0)$.

Similar to the proof of $R_2(\beta^{(r)})$, we can prove that
\begin{align}
&\|R_{31}(\beta^{(r)})\|=o_P\big(\sqrt{n}\,\big)\ ;\\
&\|R_{32}(\beta^{(r)})\|=o_P\big(\sqrt{n}\,\big)\ ;
\end{align}
\begin{equation}
  \sup_{\beta^{(r)}\in {\cal B}_n}\|R_{33}(\beta^{(r)})-U_2(\beta^{(r)}_0)\| = o_P(\sqrt{n}\,).
\end{equation}
For $R_{34}(\beta^{(r)})$, by Taylor expansion of $\beta^{(r)}-\beta_0^{(r)}$ with a suitable mean $\bar{\beta}^{(r)}\in{\cal B}_n$ and $\bar{\beta}=\bar{\beta}(\bar{\beta}^{(r)})$, we get $R_{34}(\beta^{(r)})= M_{1}\bigl(\beta^{(r)},\bar{\beta}^{(r)}\bigr)+M_{2}\big(\beta^{(r)},\bar{\beta}^{(r)}\big)$
\begin{align*}
  M_{1}\bigl(\beta^{(r)},\bar{\beta}^{(r)}\bigr)=&\sum^{n}_{i=1} (\varepsilon_i+\varepsilon^{-1}_i)g'(X_i^\tau\beta_0)\big[\hat{g'}(X_i^\tau\bar{\beta};\bar{\beta},h_1)-g'(X_i^\tau\beta_0)\big]\\
  &\quad \times J_{\beta^{(r)}}^\tau X_iX_i^\tau J_{\bar{\beta}^{(r)}}\big(\beta^{(r)}-\beta_0^{(r)}\big)\\
 M_{2}\big(\beta^{(r)},\bar{\beta}^{(r)}\big)=& \sum^{n}_{i=1} (\varepsilon_i+\varepsilon^{-1}_i)g'^2(X_i^\tau\beta_0) J_{\beta^{(r)}}^\tau X_iX_i^\tau J_{\bar{\beta}^{(r)}}\big(\beta^{(r)}-\beta_0^{(r)}\big)
\end{align*}

By Lemma A.1 and the law of large numbers, we obtain
that
\[
 \sup_{\beta^{(r)},\bar{\beta}^{(r)}\in{\mathcal B}_n}\|M_{1}\bigl(\beta^{(r)},\bar{\beta}^{(r)}\bigr)\|=o_P\big(\sqrt{n}\,\big)
\]

and
\[
 \sup_{\beta^{(r)},\bar{\beta}^{(r)}\in{\mathcal B}_n}\|M_{2}\bigl(\beta^{(r)},\bar{\beta}^{(r)}\bigr)-nV\big(\beta^{(r)}-\beta_0^{(r)}\big)\| = o_P\big(\sqrt{n}\,\big).
\]
where $V=E\{(\varepsilon_i+\varepsilon^{-1}_i)g'(X_i^\tau\beta_0)^2 J_{\beta^{(r)}_0}^\tau X_iX_i^\tau J_{\beta^{(r)}_0}\}$

Therefore, we have
\begin{equation}
\sup_{\beta^{(r)}\in{\cal B}_n}\|R_{34}(\beta^{(r)})-n
V\big(\beta^{(r)}-\beta_0^{(r)}\big)\| = o_P\big(\sqrt{n}\,\big).
\end{equation}

The proof of $R_{35}(\beta^{(r)})$ is similar to those in the proof of Lemma 5 of [\citenum{Wang:2010dj}].
Therefor we can get
\begin{equation}
\sup_{\beta^{(r)}\in {\mathcal B}_n}\|R_{35}(\beta^{(r)})\| = o_P\big(\sqrt{n}\,\big)
\end{equation}



Substituting (A.2)--(A.15) into (A.1), we prove Lemma A.4
\end{proof}
\normalsize

{\bf Proof of Theorem~\ref{thLinear}}

This follows from arguments similar to the proof Theorem~2 of [\citenum{Wang:2010dj}] and is omitted.

\section*{References}

\bibliographystyle{elsarticle-num-names}

\end{document}